\documentclass[jsl]{asl}
\usepackage{url}
\usepackage{graphics}
\newtheorem{theorem}{Theorem}
\newtheorem{corollary}[theorem]{Corollary}
\newtheorem{lemma}[theorem]{Lemma}
\numberwithin{equation}{section}

\DeclareMathOperator{\KS}{\mathit{C}}
\DeclareMathOperator{\KP}{\mathit{K}}

\newcommand{\logg}{\log^{(2)}}
\newcommand{\loggg}{\log^{(3)}}
\newcommand{\la}{\langle}
\newcommand{\ra}{\rangle}

\begin{document}

\title[Complexity of complexity and maximal $C$ and $K\,$-\,complexity]
{Complexity of complexity and strings with maximal plain and prefix Kolmogorov complexity}

\author{B. Bauwens}
\revauthor{Bauwens, Bruno}
\address{
LORIA, Universit\'e de Lorraine, 
615-B248 Rue du Jardin Botanique,
54506 Vand\OE vre-l\`es-Nancy, France
\urladdr{www.bcomp.be}.
}

\author{A. Shen}
\revauthor{Shen, Alexander}
\address{LIRMM CNRS \& University of Montpellier, 2.
 UMR 5506 - CC477, 161 rue Ada, 34095 Montpellier Cedex 5, France. On leave from IITP RAS, Moscow. }
\thanks{%
    Bruno Bauwens is supported by the Portuguese science foundation FCT
    (SFRH/BPD/75129/2010); and partially supported by $CSI^2$ (PTDC/EIA-CCO/099951/2008) and NAFIT ANR-08-EMER-008-01 projects.
    Alexander Shen is supported by NAFIT ANR-08-EMER-008-01 project.   
    Part of this paper was published in the proceedings of the 39-th International Conference on
    Automata, Languages and Programming (ICALP2012). 
    The authors are grateful to Elena Kalinina and Nikolay Vereshchagin for useful discussions.
    }

\begin{abstract}
   P\'eter G\'acs showed~(G\'acs 1974) that for every $n$ there exists a bit string $x$ of length $n$
   whose plain complexity $\KS(x)$ has almost maximal conditional complexity relative to $x$, i.e.,
   $\KS(\KS(x)|x)\ge \log n - \logg n - O(1)$. (Here $\logg i =\log\log i$.) 
   Following Elena Kalinina~(Kalinina 2011), we provide a simple game-based proof of this result; modifying her
   argument, we get a better (and tight) bound $\log n - O(1)$.  We also show the same bound for
   prefix-free complexity.
   
   \smallskip
   Robert Solovay showed~(Solovay 1975) that
   infinitely many strings $x$ have maximal plain complexity but not maximal prefix
   complexity (among the strings of the same length): for some $c$ there exist infinitely many $x$ such that  $|x| - \KS(x) \leq c$ and
   $|x| + \KP(|x|) - \KP(x) \geq \logg |x| - c\loggg |x|$.
   In fact, the results of Solovay and G\'acs are closely related.
   Using the result above, we provide a short proof for Solovay's result. We also generalize it 
   by showing that for some $c$ and for all $n$ there are strings $x$ of length $n$ 
   with $n - \KS(x) \leq c$
   and 
   $$ 
    n + \KP(n) - \KP(x) \geq \KP(\KP(n)|n) - 3\KP(\KP(\KP(n)|n)|n) - c.
   $$ 
   We also prove a close upper bound $\KP(\KP(n)|n) + O(1)$.
      \smallskip
   
   Finally, we provide a direct game proof for Joseph Miller's
   generalization~(Miller 2006) of the same Solovay's theorem: if a co-enumerable set (a
   set with c.e. complement) contains for every length a string of this length, then it contains
   infinitely many strings $x$ such that 
   $$
    |x| + K(|x|) - K(x) \geq  \logg |x| - O(\loggg |x|).
   $$
\end{abstract}

\maketitle  

\section*{Introduction}

Plain Kolmogorov complexity $\KS(x)$ of a binary string $x$ was defined in~\cite{Kolmogorov65} as the
minimal length of a program that computes $x$. 
(See, e.g.,~\cite{GacsNotes,LiVitanyi,ShenIntro,ZvonkinLevin} for more details.) 
It was clear from the beginning that this complexity function
is not computable: no algorithm can compute $\KS(x)$ given $x$. In~\cite{complexityOfComplexity} (see also~\cite{GacsNotes,LiVitanyi}) a stronger
non-uniform version of this result was proven: for every $n$ there exists a string $x$ of length $n$
such that conditional complexity $\KS(\KS(x)|x)$, i.e., the minimal length of a program that maps
$x$ to $\KS(x)$, is at least $\log n - \logg n - O(1)$. (If the complexity function were computable, this
conditional complexity would be bounded.) 

In Section~\ref{sec:gacs} we revisit this classical result and improve it a bit by removing the
$\logg n$ term.\footnote{
    Note added in proof: alternatively, this improvement can also be shown using~\cite[Theorem~3.1]{StephanEnumerationsK} 
    (and Theorem~5.1 for prefix complexity).
  }
No further improvement is possible because $\KS(x) \leq n + O(1)$ for every string $x$ of length $n$, therefore
$\KS(\KS(x)|x) \leq \log n + O(1)$ for all such~$x$. 
We also prove that we can guarantee $\KS(x) \ge n/2$ (in addition to $\KS(\KS(x)|x)\ge\log n-O(1)$),
which was (in weaker form) 
conjectured by Robert~Solovay and Gregory~Chaitin, and mentioned as Conjecture 3.14.6 on p.~145 in \cite{Downey}.

We use a game technique that was developed by Andrej
Muchnik (see~\cite{KolmogorovGames,muchnikGame,VerSurvey}) and turned out to be useful in many cases.
Recently Elena Kalinina (in her master thesis~\cite{kalinina}) used it to provide a proof of G\'acs'
result. We use a more detailed analysis of essentially the same game to get a better bound.

\smallskip
In Section~\ref{sec:solovay} we use this improved bound to provide a simple proof of an old result
due to Solovay. The complexity $\KS(x)$ of an $n$-bit string $x$ never exceeds $n+O(1)$, and
for most $n$-bit strings $x$ the value of $\KS(x)$ is close to $n$. Such strings may be called
``$C$-random''. There is another version of complexity, called prefix complexity, where the programs
are assumed to be self-delimiting (see~\cite{GacsNotes,LiVitanyi,ShenIntro} for details). For an
$n$-bit string $x$ its prefix complexity $\KP(x)$ does not exceed $n+\KP(n)+O(1)$, and for most
$n$-bit strings $x$ the value of $\KP(x)$ is close to $n+\KP(n)$.  Such strings may be called
$K$-random\footnote{In \cite{Downey}, such strings are called ``strongly $K$-random'', in contrast 
    to the ``weakly $K$-random'' strings $x$, which only satisfy $\KP(x) \ge |x|-O(1)$.}. 

A natural question arises: how ``$C$-randomness'' and ``$K$-randomness'' are related? This question
was studied by Solovay who proved that $K$-randomness implies $C$-ran\-dom\-ness but not vice versa
(see the unpublished notes~\cite{Solovay} and its exposition in~\cite{Downey}). More precisely, consider the ``randomness deficiencies'' $d_C(x)=|x|-\KS(x)$ and $d_K(x)=|x|+\KP(|x|)-\KP(x)$. Solovay proved that:
\begin{enumerate}
  \item\label{sol1} $d_C(x) \le O(d_K(x))$;
  \item\label{sol2} the reverse statement can be proved with additional error term: 
    $$
     d_K(x)\le O(d_C(x))+\logg n
    $$  
    for every $n$-bit string $x$; 
  \item\label{sol3} the error term cannot be deleted: there exists a constant $c$ and infinitely many
    strings $x$ such that $d_C(x)\le c$ but 
    $$ 
     d_K(x)\ge \logg |x| -O(\loggg |x|).
    $$ 
\end{enumerate}

Using the result of Section~\ref{sec:gacs}, we provide a short proof for statement (3), the most difficult one, even in a stronger form where $O(\loggg |x|)$ is replaced by $O(1)$. Then we prove a stronger statement about strings of fixed length $n$, with close lower and upper bounds:
\begin{itemize}
  \item $d_K(x)\le O(d_C(x))+\KP(\KP(n)|n)$ for every $n$-bit string $x$;
  \item for some constant $c$ and for every $n$ there exist a string $x$ of length $n$ 
    such that $d_C(x) \leq c$ and 
    $d_K(x)\geq \KP(\KP(n)|n) - 3\KP(\KP(\KP(n)|n)|n) - c.$
\end{itemize} 
It is stronger because the result of Section 1 then allows us to choose $n$ in such a way that $\KP(\KP(n)|n)=\logg n+O(1)$; this choice also makes the other term $O(1)$. 

\smallskip
 Finally, in Section~\ref{sec:miller} we give another example of game technique by presenting a
 simple proof of a different generalization of Solovay's result. This generalization is due to
 Miller \cite{MillerContrasting}: every co-enumerable set (a set with c.e. complement)
that contains a string of every length, contains infinitely many $x$ such that
    $$
d_K(x)\ge \logg|x|-O(\loggg|x|).
     $$

\section{Complexity of complexity can be high}\label{sec:gacs}

\begin{theorem}\label{th:main}
  There exist some constant $c$ such that for every $n$ there exists a string $x$ of length $n$ such that $\KS(\KS(x)|x) \geq \log n - c$.  
\end{theorem}

To prove this theorem, we first define some game and show a winning strategy for the game. (The connection between the game and the statement that we want to prove will be explained later.)

\subsection{The game}

Game $G_n$ has parameter $n$ and is played on a rectangular board divided into cells. The board has $2^n$ columns and $n$ rows numbered $0,1,\ldots,n-1$ (the bottom row has number $0$, the next one has number $1$ and so on, the top row has number $n-1$), see Fig.~\ref{fig:board}.

Initially the board is empty. Two players: White and Black, alternate their moves. At each move,
a player can pass or place a token (of his color) on the board.  The token can not be moved or removed
afterwards.  Also Black may blacken some cell instead. Let us agree that White starts the game
(though it does not matter).

The position of the game should satisfy some restrictions; the player who violates these restrictions, loses the game immediately. Formally the game is infinite, but since the number of (non-trivial) moves is a priori bounded, it can be considered as finite, and the winner is determined by the last (limit) position on the board.

\emph{Restrictions}: (1)~each player may put at most $2^i$ tokens in row $i$ (thus the total number of black and white tokens in a row can be at most $2^i + 2^i$); (2)~in each column Black may blacken at most half of the cells.

We say that a white token is \emph{dead} if either it is on a blackened cell or has a black token in the same column strictly below it.

\emph{Winning rule}: Black wins if he killed all white tokens, i.e., if each white token is dead in the final position.

\begin{figure}[h]
\begin{center}
\includegraphics{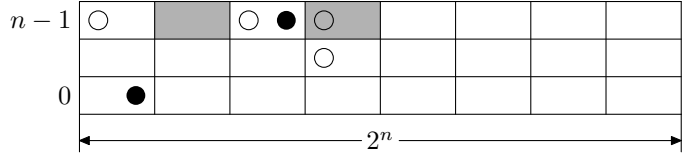}
\end{center}
\caption{Game board}\label{compcomp-1.mps}
\label{fig:board}
\end{figure}

For example, if the game ends in the position shown at Fig.~\ref{compcomp-1.mps}, the restrictions are
not violated (there are $3 \leq 2^2$ white tokens in row $2$ and $1\le 2^1$ white token in row $1$, as
well as $1 \leq 2^2$ black token in row $2$ and $1\le 2^0$ black token in row~$0$). Black loses
because the white token in the third column is not dead: it has no black token below and the cell is
not blackened. (There is also one living token in the fourth column.)

\subsection{How White can win}

The strategy is quite simple. White starts by placing a white token in an upper row of some column
and waits until Black kills it, i.e., blackens the cell or places a black token below. In the first
case White puts a token directly below it, and waits again. Since Black has no right to make all cells in a column black (at most half may be blackened), at some point he will be forced to place a black token below the white token in this column. After that White switches to some other column. (The ordering of columns is not important; we may assume that White moves from left to right.)

Note that when White switches to a next column, it may happen that there is a black token in this column or some cells are already blackened.  If there is already a black token, White switches again to the next column; if some cell is blackened, White puts her token in the topmost white (non-blackened) cell. 

This strategy allows White to win. Indeed, Black cannot place his tokens in all the columns due to the restrictions (the total number of his tokens is $\sum_{i=0}^{n-1} 2^i = 2^n -1$, which is less than the number of columns). White also cannot violate the restriction for the number of her tokens on some row $i$: all dead tokens have a black token strictly below them, so the number of them on row $i$ is at most $\sum_{j=0}^{i-1} 2^j = 2^i - 1$, hence White can put an additional token.

In fact we may even allow Black to blacken all the cells except one in each column, and White will still win, but this is not needed (and the $n/2$ restriction will be convenient later).

\subsection{Proof of G\'acs' theorem}

Let us show that for each $n$ there exists a string $x$ of length $n$ such that $\KS(\KS(x|n)|x)\ge
\log n -O(1)$.  Note that here $\KS(x|n)$ is used instead of $\KS(x)$; the difference between these
two numbers is $O(\log n)$ since $n$ can be described by $\log n$ bits, so the difference between
the complexities of these two numbers is $O(\logg n)$.

Consider the following strategy for Black (assuming that the columns of the table are indexed by
strings of length $n$):

\begin{itemize}
\item Black blackens the cell in column $x$ and row $i$ as soon as he discovers that $\KS(i|x)<\log
  n-1$. (The constant $1$ guarantees that less than half of the cells will be blackened.) Note that
  Kolmogorov complexity is an upper semicomputable function, and Black approximates it from above,
  so more and more cells are blackened.

\item Black puts a black token in a cell $(x,i)$ when he finds a program of length $i$ that produces
  $x$ with input $n$ (this implies that $\KS(x|n) \leq i$).  Note that there are at most $2^i$
  programs of length $i$, so Black does not violate the restriction for the number of tokens on any
  row $i$.  
\end{itemize}

Let White play against this strategy (using the strategy described above). Since the strategy is
computable, the behavior of White is also computable. One can construct a decompressor
$V$ for the
strings of length $n$ as follows: each time White puts a token in a cell $(x,i)$, a program of length
$i$ is assigned to~$x$. By White's restriction, no more than $2^i$ programs need to be assigned. By
universality, a white token on cell $(x,i)$ implies that $\KS(x|n) \leq i + O(1)$. If White's token is
alive in $(x,i)$, there is no black token below, so $\KS(x|n)\ge i$, and therefore
$\KS(x|n)=i+O(1)$. Moreover, for a living token, the cell $(x,i)$ is not blackened, so $\KS(i|x)\ge
\log n -1$. Therefore, $\KS(\KS(x|n)|x)\ge \log n-O(1)$.

\textbf{Remark}: the construction also guarantees that $\KS(x|n)\ge n/2-O(1)$ for that~$x$. 
(Here the factor $1/2$ can be replaced by any $\alpha<1$ if we change the rules of the game accordingly.)
Indeed, according to White's strategy, he always plays in the highest non-black cell of some column, 
and at most half of the cells in a column can be blackened, therefore no white tokens appear in the
lower half of the board.

\subsection{Modified game and 
proof of Theorem~\protect{\ref{th:main}}}

Now we need to get rid of the condition $n$ and show that for every $n$ there is some $x$ such that
$\KS(\KS(x)|x)\ge \log n - O(1)$. Imagine that White and Black play simultaneously all the games
$G_n$. Black blackens the cell $(x,i)$ in game $G_{|x|}$ when he discovers that $\KS(i|x)<\log |x| - 1$,
as he did before, and puts a black token in a cell $(x,i)$ when he discovers an \emph{unconditional}
program of length $i$ for $x$. If Black uses this strategy, he satisfies the stronger restriction:
the total number of tokens in row $i$ \emph{on all boards} is bounded by $2^i$. 

Assume that White uses the described strategy on each board. What can be said about the total number
of white tokens in row $i$? The dead tokens have black tokens strictly below them and hence the total
number of them does not exceed $2^i - 1$. On the other hand, there is at most one living white token on
each board. We know also that in $G_n$ white tokens never appear below row $n/2-1$, so the number of
alive white tokens does not exceed $2i + O(1)$. Therefore we have $O(2^i)$ white tokens on the $i$-th row
in total.

For each $n$ there is a cell $(x,i)$ in $G_n$ where White wins in $G_n$. Hence, $\KS(x)<i+O(1)$
(because of the  
property just mentioned and the computability of White's behavior), $\KS(x)\ge i$ and
$\KS(i|x)\ge \log n-1$ (by construction of Black's strategies and the winning condition).
Theorem~\ref{th:main} is proven.

\subsection{Version for prefix complexity}

\begin{theorem}\label{th:prefix-max}
There exist some constant $c$ such that for every $n$ there exists a string $x$ of length $n$ such
that $\KS(\KP(x)|x) \geq \log n-c$ and $\KP(x) \geq n/2-c$. This also implies that $\KP(\KP(x)|x) = \log n+O(1)$.
\end{theorem}

The proof of $\KS(\KP(x)|x) \geq \log n - c$ goes in the same way. Black places a token in cell
$(x,i)$ if some program of length $i$ for a prefix-free (unconditional) machine computes $x$ (and
hence $\KP(x) \leq i$) and blackens the cell if $K(i|x) < n-1$; White uses the same strategy as described above. The sum of $2^{-i}$ for all
black tokens is less than $1$ (Kraft inequality); some white tokens are dead, i.e., strictly above black ones, and for
each column the sum of $2^{-j}$ over these tokens $(x,j)$, does not exceed $\sum_{j>i}^n 2^{-j} <
2^{-i}$. Hence the corresponding sum for all dead white tokens is less than $1$; for the rest the sum
is bounded by $\sum_n 2^{-n/2 + 1}$, so the total sum is bounded by a constant, and we conclude that for
the token in the winning column $x$ the row number is $\KP(x)+O(1)$, and this cell is not blackened.

It remains to note that $\KP(\KP(x)|x)$ is greater than $\KS(\KP(x)|x)\ge \log n - O(1)$ for $x$ of
length $n$; on the other hand,  $n/2 \le \KP(x)\le 2n+O(1)$, so the length of $\KP(x)$ (in binary) 
is $\log n + O(1)$, and the conditional prefix complexity of a string given its length is bounded by
the length, hence $\KP(\KP(x)|x) \le \KP(\KP(x)|\log n) + O(1) \le \log n + O(1)$.

\medskip
\textbf{Remark}:
 In fact, $\KP(\KP(x)|x) \le \log |x| + O(1)$ for all $x$ (this will be useful in the next section). In general,
 if $z \le O(n)$, then $\KP(z|\log n)\le \log n+O(1)$, because we may add leading zeros to the
 binary representation of $z$ up to length $\log n+O(1)$, and the prefix complexity of a string given its
 length does not exceed the length. (Note that for $z=\KP(x)$ and $n=|x|$ we have $z \le O(n)$, and
 $\KP(\KP(x)|x)\le \KP(\KP(x)|n)\le \KP(\KP(x)|\log n)$.)

\section{Strings with maximal plain and prefix complexity}
  \label{sec:solovay}

In this section we provide a new proof and a generalization for Solovay's result
mentioned in the introduction. For completeness we first reproduce a proof of the simple upper bound for $d_C(x)$ in terms of $d_K(x)$. 

\begin{theorem}[Solovay \cite{Solovay}]
  \label{th:KrandImpliesCrand} $$d_C(x)\le O(d_K(x)).$$
\end{theorem}
\begin{proof}
Assume that $d_C(x)$ is large for some $x$ of length $n$; we need to show that $d_K(x)$ is almost as
large. Let $d_C(x)$ be equal to some $c$, and $p$ be a plain program for $x$. Let $\hat c$ and
$\hat{n}$ be the self-delimiting programs for $c$ and $n$ of length $O(\log c)$ and $\KP(n)$. Then
$\hat{c}\hat{n}p$ is a self-delimiting program for $x$ that gives prefix deficiency $c-O(\log c)$.
\end{proof}

As we have mentioned, the reverse statement is not true: $d_K(x)$ can be big even if $d_C(x)$ is small. However, there exists an upper bound for $d_K$ in terms of $d_C$ and other complexities:

\begin{theorem}\label{th:solovayOptimal}
 For any $x$ of length $n$ 
 $$
  d_K(x) \leq O(d_C(x))+\KP(\KP(n) |n).
 $$
\end{theorem}
     
Note that $\KP(\KP(n)|n)\le \logg n+O(1)$ (see the remark that ends the previous section), so this
bound implies the bound from~\cite{Solovay} mentioned in the introduction). 

\begin{proof} 
Let us denote $d_C(x)$ by $c$. As Levin noted, $\KS(x)=\KP(x|\KS(x))+O(1)$~\cite{levinCK} (see
also~\cite[p.~203]{LiVitanyi}). So with $O(c)$-precision we have
$$n = \KS(x) = \KP(x|\KS(x)) = \KP(x|n).$$
 Now we apply 
$$
    \KP(u,v) = \KP(u) + \KP(v|u,\KP(u))+O(1).
$$
(additivity for prefix complexity, see, e.g.,~\cite{complexityOfComplexity,GacsNotes,LiVitanyi}) for $u=n$, $v=x$: 
  $$
\KP(x) = \KP(n,x) = \KP(n) + \KP(x|n,\KP(n)),
  $$
all with $O(1)$-precision. Combining these two observations, we get
 \begin{align}
   d_K(x)&=n + \KP(n) - \KP(x)=\nonumber
   \\ &= (\KP(x|n) + O(c))+ \KP(n) - (\KP(n) + \KP(x|n,K(n)))= \nonumber \\
   &= \KP(x|n) - \KP(x|n,\KP(n)) + O(c).\label{eq:gap}
 \end{align}
 It is easy to see that $\KP(x|n) \leq \KP(x|n,\KP(n)) + \KP(\KP(n)|n)+O(1)$, so $d_K(x)$ is 
 bounded by $\KP(\KP(n)|n)+O(c)$.
\end{proof}

\textbf{Remark}: With essentially the same proofs, we can replace terms $O(d_K(x))$ and $O(d_C(x))$ by 
$d_K(x) + O(\log d_K(x))$ and $d_C(x) + O(\log d_C(x))$ in Theorems~\ref{th:KrandImpliesCrand} and~\ref{th:solovayOptimal}.

\bigskip
The following theorem shows that \emph{for all $n$} the second term in the bound of Theorem~\ref{th:solovayOptimal} is unavoidable (up to $O(\log \KP(\KP(n)|n))$ precision).

\begin{theorem} \label{th:solovay2allN}
  For some $c$ and all $n$ there exists a string $x$ of length $n$ such that $d_C(x) \leq c$, and
  \[
  d_K(x) \geq \KP(\KP(n)|n) - 3\KP(\KP(\KP(n)|n)|n) - c.
  \]
\end{theorem}

As we have said in the introduction, we can combine this result with Theorem~\ref{th:prefix-max} to obtain Solovay's result as corollary, even without $\loggg$-term:

\begin{corollary}\label{cor:solovay2}
  There exists a constant $c$ and infinitely many $x$ such that $d_C(x) \le c$ and 
  $d_K(x)\ge \log^{(2)} |x|-c$.  
\end{corollary}

Before proving Theorem \ref{th:solovay2allN}, we prove the corollary directly.
\begin{proof}
  First we choose $n$,  the length of string $x$, in such a way that
  $$\KP(\KP(n)|n)=\logg n +O(1)$$  
  and $\KP(n)\ge (\log n)/2 - O(1)$ (Theorem~\ref{th:prefix-max}).  We know already from
  equation~\eqref{eq:gap} that for a string $x$ with $\KS$-deficiency $c$ the value of
  $\KP$-deficiency is $O(c)$-close to $\KP(x|n)-\KP(x|n,\KP(n))$.  In other words, adding $\KP(n)$ to the condition $n$ in $\KP(x|n)$ should decrease the complexity, so let us include $\KP(n)$ in $x$ somehow. We also
  have to guarantee maximal $\KS$-complexity of $x$. This motivates the following choice:
  
 \begin{itemize}
   \item choose $r$ of length $n - \logg n$ such that $\KP(r|n,\KP(n)) \geq |r|$. 
     Note that this implies $\KP(r|n,\KP(n)) = |r|+O(1)$, since the length of $r$ is determined by the condition.
   \item Let $x = \la K(n) \ra r$, the concatenation of $K(n)$ (in binary) with $r$.
     Note that $\la K(n) \ra$ has at most $\logg n + O(1)$ bits for every $n$, and by choice of $n$
     has at least $\logg n - O(1)$ bits, hence $|x| = n + O(1)$.
 \end{itemize}
 As we have seen (looking at equation~\eqref{eq:gap}), it is enough to show that 
 $$\KP(x| \KP(n),n) \le n - \logg n$$
  and $\KP(x|n)=n$ (the latter equality implies $\KS(x)=n$, as
 Levin has noted\footnote{We already mentioned the equation $\KS(x)=\KP(x|\KS(x))+O(1)$, so $\KS(x)$
 is a fixed point of the function $i\mapsto \KP(x|i)$ up to $O(1)$-precision. Since this function
 changes logarithmically slow compared to $i$, the reverse statement is also true: if $\KP(x|i)=i$
 with some precision $d$, then $i = \KS(x)$ with $O(d)$-precision.}); all the equalities here
 and below are up to $O(1)$ additive term.
 \begin{itemize}
   \item 
     Knowing $n$, we can split $x$ in two parts $\langle \KP(n)\rangle$ and $r$. 
     Hence, $\KP(x|\KP(n),n) = \KP(\KP(n),r|n,\KP(n))$, and this equals $\KP(r|n,\KP(n))$, i.e.,  
     $n - \logg n$ by choice of~$r$. 
   \item 
    To compute $\KP(x|n)$, we use additivity:
     \begin{equation*}
      \KP(x | n) = \KP(\KP(n),r|n) = \KP(\KP(n)|n) + \KP(r | \KP(n), \KP(\KP(n)|n), n). 
     \end{equation*}
     By choice of $n$, we have $\KP(\KP(n)|n)=\logg n$, and the last term simplifies to $\KP(r |
     \KP(n), \logg n, n)$, and this equals $\KP(r | \KP(n), n) = n - \logg n$ by choice of~$r$.
     Hence $\KP(x|n) = \logg n + (n - \logg n) = n$.
 \end{itemize}
\end{proof}

\textbf{Remark:}
 One can ask how many strings are suitable for Corollary \ref{cor:solovay2}.
By  Theorem~\ref{th:solovayOptimal}, the length $n$ of such a string must satisfy 
$\KP(\KP(n)|n) \geq \logg n - O(1)$.  
By Theorem~\ref{th:prefix-max}, there is at least one such $n$ for every $|n|$ (length 
of $n$ as a binary string). Hence such $n$ can be found within exponential intervals. 

Then one can ask (for some $n$ with this property) how many strings $x$ of length $n$ are suitable for Corollary \ref{cor:solovay2}. 
By a theorem of Chaitin \cite{LiVitanyi}, there are at most $O(2^{n-k})$ strings of length~$n$ with
$\KP$-deficiency $k$, hence we can have at most $O(2^{n-\logg n})$ such strings. It turns out
that at least a constant fraction of them is suitable for Corollary \ref{cor:solovay2}.
To show this, note that in the proof 
every different $r$ of length $|n|-\logg n + O(1)$ 
leads to a different $x$. For $r$ we need $\KP(r|n,\KP(n)) \geq
|r|-O(1)$, and  there are $O(2^{n - \logg n})$ such~$r$.

\medskip
The corollary is proved, and we proceed to the 
\smallskip

\begin{proof}[Proof of Theorem \ref{th:solovay2allN}]
  In the proof above, in order to obtain a large value $\KP(x|n)-\KP(x|n,\KP(n))$, 
  we incorporated $\KP(n)$ directly in $x$ (as $\la \KP(n) \ra$). 
  To show that $C(x) = K(x|n)+O(1)$ is large, we used that the length of
  $\la \KP(n) \ra$ equals $\KP(\KP(n)|n) + O(1)$. For arbitrary $n$
  this trick does not work, 
  but we can use a shortest program for $\KP(n)$ given $n$ (on a plain machine) instead of $\la\KP(n)\ra$. 
  For every $n$, we construct $x$ as follows:
  \begin{itemize}
    \item  
      let $q$ be a shortest program that computes $\KP(n)$ from $n$ on a {\em plain} machine (if there are several shortest
      programs, we choose the one that appears first, so it can be reconstructed from $n$ and $\KP(n)$). 
      Note that $|q| = \KS(\KP(n)|n)= \KS(q|n) + O(1)$ (remember that a shortest program is always incompressible). By Levin's result (conditional version: $\KS(u|v)=\KP(u|v,\KS(u|v))$), the last term also equals
      $\KP(q|n,|q|)+O(1)$; 
    \item 
      let $r$ be a string of length $n - |q|$ such that $\KP(r|n,\KP(n),q) \geq |r|$.
      This implies $\KP(r|n,\KP(n),q) = |r|+O(1)$, since the length of $r$ is determined by the
      condition.
    \item 
      now we define $x$ as the concatenation $qr$.
  \end{itemize}
  
 Now the proof goes as follows. We have to prove two things (together they obviously imply the statement of Theorem~\ref{th:solovay2allN}):
 \begin{itemize} 
 \item that $\KS(x) = n + O(1)$,  and $d_K(x)\ge |q|-\KP(|q|\,|n) + O(1)$.
 \item that $\KP(\KP(n)|n) - 3\KP( \KP(\KP(n)|n) |n) \leq \KS(\KP(n)|n) - \KP(\KS(\KP(n)|n)|n) + O(1).$
  \end{itemize}
The second part is a special case of the following  
 \begin{lemma}\label{lem:plainPrefixHelp}
 $\KP(a|b) - 3\KP(\KP(a|b)|b) \leq \KS(a|b) - \KP(\KS(a|b)|b)+O(1)$
\end{lemma}
\noindent
for $a=\KP(n)$ and $b=n$. The proof of this lemma will be given after we finish the rest of the proof.
  
 For the first part we follow the same structure as above. Using equation~\eqref{eq:gap}, we see that it is enough to show that with $O(1)$-precision (we omit $O(1)$-terms in the sequel) we have $\KP(x| \KP(n),n) \leq n - |q| + \KP(|q|\,|n)$ and $\KP(x|n)=n$ (the latter equality implies $\KS(x)=n$). Let us prove these two statements:
 
\begin{itemize}
   \item 
     Knowing $|q|$, we can split $x$ in two parts $q$ and $r$. 
     Hence, $\KP(x|\KP(n),n,|q|) = \KP(q,r|n,\KP(n),|q|)$. 
     Given $n, \KP(n), |q|$ we can search for a program of length $|q|$ 
     that on input $n$ outputs $\KP(n)$; the first one is $q$.
     Hence, $$\KP(q,r|n,\KP(n),|q|) = \KP(r|n,\KP(n),|q|)=n-|q|$$ 
     (the last equality is due to the choice of $r$), and
     therefore $$\KP(x|\KP(n),n) \leq n - |q| + \KP(|q|\,|n).$$

   \item 
     For $\KP(x|n)$ we use additivity:
     \begin{equation*}
      \KP(x | n) \ge \KP(x|n, |q|) = \KP(q,r|n, |q|) = \KP(q|n, |q|) + \KP(r | q, \KP(q|n,|q|), n). 
     \end{equation*}
     By choice of $q$ we have $\KP(q|n, |q|) = |q|$.
     The last term is $\KP(r | q, |q|, n)$ and is equal to $\KP(r|q,n) = n - |q|$ by choice of~$r$.
     Hence, $\KP(x|n) \ge |q| + (n - |q|) = n$.  Since $x$ is an $n$-bit string, we have also 
     $\KP(x|n)\le n$.
 \end{itemize}
\end{proof}

\noindent
Theorem~\ref{th:solovay2allN} is proved except for the proof of Lemma~\ref{lem:plainPrefixHelp}, which we give now.

\begin{proof}
The condition $b$ is used everywhere, so the statement is a conditional version of the inequality 
     $$
\KP(a) - 3\KP(\KP(a)) \leq \KS(a) - \KP(\KS(a))+O(1).
     $$
As usually, the proof of the conditional version follows the unconditional one, so we 
consider the unconditional version for simplicity.     
 Note that $\KP(a) - \KS(a) \leq \KP(\KS(a))$. 
 Indeed, every program $p$ for plain machine can be converted to a self-delimiting one by
 adding a self-delimiting description of $|p|$ before $p$.
 Hence it remains to show that
 $2\KP(\KS(a)) \leq 3\KP(\KP(a))+O(1)$. This follows from another Solovay's result from~\cite{Solovay} (see also~\cite{Downey}) which says that 
 \[
 \KP(a) - \KS(a) = \KP(\KP(a)) + O(\KP(\KP(\KP(a)))).
 \]
From this result we conclude that
 \[
  |\KP(\KP(a)) - \KP(\KS(a))| \leq O(\log \KP(\KP(a))) ,
 \]
 and this is enough for our purpose.
\end{proof}

\section{Game-theoretic proof of Miller's theorem}
\label{sec:miller}

In this section we provide a simple game-based proof of a result due to
Miller~\cite{MillerContrasting}; as we have seen in the introduction, this result implies that
$C$-randomness differs from $K$-randomness.  (The original proof in~\cite{MillerContrasting} uses a different scheme that involves the Kleene fixed-point theorem.)

\begin{theorem}[J.~Miller]\label{th:solovay2Gen}
  For any co-enumerable set $Q$ of strings that contains a string of every length, 
  there exist infinitely many $x$ in $Q$ such that $d_K(x) \ge  \logg |x| - O(\loggg |x|)$. 
\end{theorem}

Solovay's result follows by choosing $Q$ to be the set of strings $x$ such that $d_C(x) \le c$ 
for large enough $c$ (then $Q$ contains strings of all lengths); this set is co-enumerable. One
can also conclude that the set of strings $x$ with $d_K(x)<c$ is not co-enumerable for large enough
$c$ (when this set contains strings of all lengths).  One may also observe that because of Theorem~\ref{th:solovayOptimal}, this result also implies a weak form of G\'acs' theorem: there exist
infinitely many $x$ such that $\KP(\KP(x)|x) \ge \log |x| - O(\logg |x|)$.

\medskip
\begin{proof} Let us consider the following game specified by a natural number $C$ and a
finite family of disjoint finite sets $S_1,\dots,S_N$. During the game each element  
$s\in S=\cup_{j=1}^N S_j$ is labeled by two non-negative rational numbers $A(s)$ and $B(s)$ called ``Alice's weight'' and ``Bob's weight''. Initially all weights are zeros. Alice and Bob make alternate moves.
On each move each player may increase her/his weight of several elements~$s\in S$.

Both players must obey the following restrictions for the total weight: 
       $$
\sum_{s\in S}A(s)\le1\quad\text{and}\quad \sum_{s\in S}B(s)\le1.
       $$
In addition, Bob must be ``fair'': for every $j$ Bob's weights of all $s\in S_j$ must be equal. That means that basically Bob assigns weights to $j\in\{1,\dots,N\}$ and Bob's weight $B(j)$ of $j$ is then evenly distributed among all $s\in S_j$ so that $$ B(s)=B(j)/\#S_j $$ for all $s\in S_j$. Alice does not need to be fair.

This extra requirement is somehow compensated by allowing Bob to ``disable'' certain $s\in S$ (this
does not decrease the size of $S$). Once
an $s$ is disabled it cannot be ``enabled'' any more. Alice cannot disable or enable anything. For
every $j$, Bob is not allowed to disable \emph{all} $s\in S_j$: every set $S_j$ should contain at least one element that is enabled (=not disabled).

The game is infinite. Alice wins if at the end of the game (or, better to say, in the limit) there
exists an enabled $s\in S$ such that
      $$
\frac{A(s)}{B(s)}\ge C.
      $$

Now we have to explain two things: why Alice has a (computable) winning strategy in the
game (with some assumptions on the parameters of the game) and why this implies Miller's theorem.

\begin{lemma}\label{lem:gameSolovayGen}
Assume that $N\ge2^{8C}$ and $\#S_j\ge 8C$ for all $j\le N$.
Then Alice has a computable winning strategy.

\end{lemma}

Let us show first why this statement implies the theorem. First we show how for a given $c$ one can find some $x\in Q$ with $d_K(x)\ge c$. (Then we look more closely on the length of this $x$ and check that indeed the statement of Theorem~\ref{th:solovay2Gen} is true.)
Consider the following values of the game parameters:
     $$
C=2^{c} \quad\text{and}\quad N= 2^{8C}=2^{2^{c+3}}
    $$
Let us take the sets of all strings of length
           $$\log 8C+1,\dots, \log 8C+N$$
as $S_1,\ldots,S_N$. 

Consider the following strategy for Bob in this game. He enumerates the complement of $Q$ and
disables all its elements. In parallel, he approximates the prefix complexity from above; once he
finds out that $K(n)$ does not exceed some~$l$, he increases the weights of all $2^n$ strings of
length~$n$ up to $2^{-l-n}$. Thus at the end of the game $B(x)=2^{-K(n)-n}$ for all $s\in S$ that
have length $n$ (i.e., for $s\in S_j$ where $j=n-\log 8C$). Note that Bob's total weight never exceeds 
its limit, since $\sum_n 2^{-\KP(n)}\le 1$.

Alice's limit weight function $x\mapsto A(x)$ is lower semi-computable given $c$, as both Alice's and Bob's strategies are computable given $c$. Therefore, since prefix complexity is equal to the logarithm of a priori probability (coding lemma),
    $$\KP(s|c)\le -\log A(s)+O(1)$$
for all $s\in S$. As Alice wins, there exists a string $s\in Q$ of some length $n\le N+\log 8C$ such that $A(s)/B(s)\ge C$, i.e., 
$$
 -\log A(s)\le -\log B(s)-c=\KP(n)+n-c.
$$
This implies that
$$
 \KP(s|c)\le \KP(n)+n-c+O(1),
$$
and
$$ 
  \KP(s) \le \KP(n)+n -c+O(\log c).
$$

Now let us look at the length of a string $s$ constructed for a given $c$. The maximal possible
length is $\log (8C)+N$, which is $O(N)$ since $N=2^{8C}$ is much bigger than $\log(8C)$. So the
length is at most $$ O\bigl(2^{2^{c+3}}\bigr).$$ 
In other terms, $c+3\ge \logg |s|-O(1)$ and the deficiency of $s$ is at least $c - O(\log c)$, which
is at least $\logg |s|-O(\loggg|s|).$
\end{proof}

It remains to prove the Lemma by showing a winning strategy for Alice.

\begin{proof}[Proof of Lemma \ref{lem:gameSolovayGen}.]  
The strategy is rather straightforward. The main idea is that playing
with one $S_i$, Alice can force Bob to spend twice more weight than she does. Then she switches to
the next $S_i$, and so on until Bob's weight is exhausted while she has solid reserves. To achieve her
goal on one set of $M$ elements, Alice assigns sequentially weights $1/2^M,
1/2^{M-1},\ldots,1/{2^1}$ and after each move waits until Bob increases his weight enough to satisfy the
game requirements, or disables the
corresponding element. Since he cannot disable all elements and is forced to use the same weights
for all elements while Alice puts more than half of the weight on the last element, Bob has factor
$M/2$ as a handicap, and we may assume that $M/2$ beats $C$-factor that Bob has in his favor.

Now the formal details. Assume first that $\#S_j=M=4C$ for all $j$ and $N=2^{M}$. (We will show
later how to adjust the proof to the case when $|S_j|\ge8C$ and $N\ge2^{8C}$.)

Alice picks an element $x_1\in S_1$ and assigns the weight $1/2^{M}$ to $x_1$.  Bob (to avoid losing
the entire game) has either to assign a weight of more than $1/C2^{M}$ to all elements in $S_1$, or
to disable $x_1$. In the second case Alice picks another element $x_2\in S_1$ and assigns a (twice
bigger) weight of $2/2^{M}$ to it. Again Bob has a dilemma: either  to increase the weight for all
elements of $S_1$  up to $2/C2^{M}$, or to disable $x_2$. In the second case Alice picks $x_3$,
assigns a weight of $4/2^{M}$ to it, and so on. (If this process continues long enough, the last
weight would be $2^{M-1}/2^M=1/2$.)

As Bob cannot disable all the elements of $S_1$, at some step $i$ the first case occurs, and Bob
assigns a weight greater than $2^{i-1}/C2^M$ to all the elements of $S_1$.  Then Alice stops playing
on $S_1$.  Note that the total Alice's weight of $S_1$ (let us call it $\beta$) is the sum of the
geometric sequence: $$
\beta=1/2^{M}+2/2^M+\ldots +2^{i-1}/2^M<2^i/2^M\le1.
       $$
Thus  Alice obeys the rules. Note that total Bob's weight of $S_1$ is more than
$M2^{i-1}/C2^M=2^{i+1}/2^M$, so it exceeds at least two times the total Alice's weight spent on
$S_1$. This implies, in particular, that Bob cannot beat Alice's weight for the last element if the
game comes to this stage (and Alice wins the game in this case.)

Then Alice proceeds to the second set $S_2$ and repeats the procedure. However this time she uses weights
     $
\alpha/2^{M},2\alpha/2^M,\dots,
    $
where $\alpha=1-\beta$ is the weight still available for Alice. Again she forces Bob to use twice
more weight than she does. Then Alice repeats the procedure for the third set $S_3$ with the
remaining weight etc.

Let  $\beta_j$  be the total weight Alice spent on the sets $S_1,\dots,S_j$, and
$\alpha_j=1-\beta_j$ the weight remaining after the first $j$ iterations. By construction, Bob's
total weight spent on sets $S_1,\dots,S_j$ is greater than $2\beta_j$, so we have $2\beta_j<1$ and
hence $\alpha_j> 1/2$. Consequently, Alice's total weight of each $S_j$ is more than $1/2^{M+1}$.
Hence after at most $N=2^{M}$ iterations Alice wins.

If the size of $S_j$ are large but different, we need to make some modifications. (We cannot use the
same approach starting with $1/2^M$ where $M$ is the size of the set: if Bob beats the first element
with factor $C$, he spends twice more weight than Alice but still a small amount, so we do not have
enough sets for a contradiction.)

However, the modification is easy. If the number of elements in $S_j$ is a multiple of $4C$ (which
is the case we use), we can split elements of $S_j$ into $4C$ groups of equal size, and treat all
members of each group $G$ as one element. This means that if the above algorithm asks to assign to
an ``element'' (group) $G$ a weight $w$, Alice distributes the weight $w$ uniformly among members of
$G$ and waits until either Bob disables all elements of the group or assigns $4C$-bigger weight to
all elements of $S_j$.

If $S_j$ is not a multiple of $4C$, the groups are not equal (the worst case is when some groups
have one element while other have two elements), so to compensate for this we need to use $8C$
instead of $4C$.

Note that excess in the number of groups (when $N$ is bigger than required $8C$) does not matter at
all, we just ignore some groups.
\end{proof}

%
%


\bibliographystyle{asl}
\bibliography{kolmogorov}

\end{document}